\documentclass{llncs}
\makeatletter
\newcommand{\@chapapp}{\relax}%
\makeatother

\usepackage[utf8]{inputenc}
\usepackage{amsmath, amssymb}
\usepackage{amsfonts}
\usepackage[title]{appendix}
\usepackage{comment}
\usepackage{hyperref}
\usepackage{lipsum}
\usepackage{mathtools}
\usepackage{enumitem}
\usepackage{tikz}

\usetikzlibrary{arrows,automata,positioning}

\allowdisplaybreaks

\newcommand{\secref}[1]{\hyperref[#1]{Section~\ref*{#1}}}

\DeclareFontFamily{OT1}{pzc}{}
\DeclareFontShape{OT1}{pzc}{m}{it}{<-> s * [1.10] pzcmi7t}{}
\DeclareMathAlphabet{\mathpzc}{OT1}{pzc}{m}{it}

\title{Temporal Logics with Language Parameters}
\author{ Jens Oliver Gutsfeld \and Markus Müller-Olm \and Christian Dielitz}
\institute{Institut für Informatik, Westfälische Wilhelms-Universität Münster \\ Einsteinstraße 62, 48149, Münster, Germany \\
\email{\{jens.gutsfeld, markus.mueller-olm, c.dielitz\}@uni-muenster.de}}

\begin{document}
	\maketitle
	\begin{abstract}
Computation Tree Logic (CTL) and its extensions CTL$^*$ and CTL$^+$  are widely used in automated verification as a basis for common model checking tools. But while they can express many properties of interest like reachability,  even simple regular properties like ``Every other index is labelled $a$'' cannot be expressed in these logics. While many extensions were developed to include regular or even non-regular (e.g. visibly pushdown) languages, the first generic framework, \textit{Extended CTL}, for CTL with arbitrary language classes was given by Axelsson et. al. and applied to regular, visibly pushdown and (deterministic) context-free languages. 
We extend this framework to CTL$^*$ and CTL$^+$ and analyse it with regard to decidability, complexity, expressivity and satisfiability. 
\end{abstract}

	\section{Introduction}
Temporal logics like LTL, CTL and CTL$^*$ are widely used for verification purposes through satisfiability checking and model checking. Their usefulness is based on their simple logical structure, their ability to capture useful properties like reachability and their low (polynomial) model checking complexity for fixed formulae. However,  as noticed already in the seminal paper by Wolper \cite{Wolper1983}, classical temporal logics are unable to express even basic regular properties like \textit{Every other index fulfills $\varphi$} for a specification $\varphi$. Several temporal logics were developed to express (subsets of) $\omega$-regular properties. Nevertheless, these formalisms are not sufficient to specify properties like \textit{A sequence of $n$ requests should be followed by $n$ acknowledgements} that are not $\omega$-regular.
While new temporal logics have also been developed for these types of properties, they are commonly ad hoc constructions for a single language class in the sense that there is no straightforward way to extend them to even larger language classes or alternatively restrict their expressive power to a fragment related to a smaller language class in order to improve the model checking complexity. Therefore, Axelsson. et. al \cite{Axelsson2010} developed a framework called \textit{Extended Computation Tree Logic} that systematically presents variants of CTL parameterised by languages of finite words. They consider regular languages, visibly pushdown languages, deterministic context-free languages and context-free languages. They extensively study the satisfiability problem and the model checking problem for different classes of system models (Kripke Transition Systems, Visibly Pushdown Systems and Pushdown Systems) and provide precise complexity results and expressivity comparisons. In particular, they prove that Extended CTL cannot express the classical LTL property of fairness regardless of the choice of the formal language parameter. As a result, they suggest to lift their framework to extended variants of CTL$^*$ and CTL$^+$, a lesser known logic that has the same expressive power as CTL, but is more succinct \cite{Emerson1982,Wilke1999}. 
In this paper, we carry out this suggestion for both logics and also develop a parameterised variant of LTL with the model checking algorithm for the latter inducing the algorithms for the former. We obtain completeness results for model checking our logics against the model classes mentioned before in almost all cases and also for the satisfiability problem. However, unlike for Extended CTL, no extension beyond the visibly pushdown languages considered is decidable for any of these new logics. 

In \autoref{section:preliminaries}, we introduce some automata models. \autoref{section:logics} introduces our extended variants of LTL, CTL$^*$ and CTL$^+$ and establishes some basic properties of them. Afterwards, we characterise LTL[U] formulae by alternating automata (\autoref{section:automataLTLU}). Later, in \autoref{section:expressivity}, we establish inclusion and separation theorems for several logics. Then, in \autoref{section:satisfiability} and \autoref{section:modelchecking}, we address the satisfiability and model checking problems for our logics. Finally, we summarise this paper and discuss related work.

\textbf{Acknowledgements} We thank Laura Bozzelli for the extended version of \cite{Bozzelli2007} and Christoph Ohrem and Bastian Köpcke for valuable discussions.

	\section{Preliminaries}\label{section:preliminaries}
Let $AP$ be a finite set of atomic propositions and $\Gamma$ be a set of actions. A Kripke Transition System (KTS) is a tuple $\mathcal{K} = (S, \rightarrow, \lambda)$ where $S$ is a set of states, $\rightarrow \subseteq S \times \Gamma \times  S$ is a transition relation and $\lambda: S \rightarrow 2^{AP}$ is a labelling function. We also write $s \xrightarrow{a} s'$ to denote $(s, a, s') \in \rightarrow$.  Every KTS has an initial state $s_0 \in S$. A path is a maximal sequence of alternating states and actions $\pi = s_0 a_0 s_1 a_1 \dots$ such that $s_i \xrightarrow{a_i} s_{i+1}$. By $\pi^i$, we denote the suffix $s_i a_i s_{i+1}a_{i+1} \dots$ and by $\pi(i)$, we denote the state $s_i$. W.l.o.g., we assume that every state has at least one successor and thus every path is infinite. For every state, we denote by $Paths(s)$ the paths starting in $s$. We also define $Paths(\mathcal{K}) = Paths(s_0)$. For a path $\pi = s_0 a_0 s_1 a_1 \dots$, we define the corresponding trace to be $\hat{\pi} = (\lambda(s_0) a_0) (\lambda(s_1) a_1) \dots$. We then lift $Paths(s)$ to $Traces(s)$ and $Paths(\mathcal{K})$ to $Traces(\mathcal{K})$ in the obvious manner. We also denote the actions $a_0 \dots a_n$ by $Actions(\pi, n)$.

A \emph{Pushdown System} is a tuple $\mathcal{P} = (P, \Gamma, \Delta, \lambda, I)$ consisting of a finite, non-empty set of control locations $P$, a finite, non-empty stack alphabet $\Gamma$, a transition relation $\Delta \subseteq (P \times \Gamma) \times \Sigma  \times  (P \times \Gamma^*)$, a labelling function $\lambda: P \times \Gamma \rightarrow 2^{AP}$, and a non-empty set of initial configurations $I \subseteq P \times \Gamma$.
We write $(p, \gamma) \xrightarrow{\sigma} (p', \omega)$ to denote $((p, \gamma),\sigma, (p', \omega)) \in \Delta$. By slight abuse of notation, we also write $(p, \gamma) \xrightarrow{\sigma} (p', \omega) \in \Delta$.
For simplicity, we assume that $\Delta$ is total, i.e. for every $(p, \gamma)$, there are $\sigma, p', \omega$ with $(p, \gamma) \xrightarrow{\sigma} (p', \omega) \in \Delta$ and that in every rule $|\omega| \leq 2$. A rule $(p, \gamma) \xrightarrow{\sigma} (p', \omega) \in \Delta$ is a \emph{call} rule if $|\omega| = 2$, an \emph{internal} rule if $|\omega| = 1$ and a \emph{return} rule if $|\omega| = 0$. 
In the following, we will always identify a PDS with the KTS consisting of its configuration graph. 

A PDS is called a \textit{Visibly Pushdown System} (VPS) if its input alphabet $\Sigma$ can be divided into disjoint partitions $\Sigma_{call}$, $\Sigma_{int}$ and $\Sigma_{ret}$ such that every rule for an input symbol from the set $\Sigma_{call}$ is a call rule (and analogously for the other two sets). An alphabet thus partitioned is called a \emph{pushdown alphabet}.
A Büchi Pushdown System (BPDS) $\mathcal{BP}$ is a PDS with an additional set $F \subseteq P$ of final states. A path of a BPDS $\mathcal{BP}$ is accepting if it visits some state $p \in F$ infinitely often. A word is accepted if there is an accepting path for it. We denote by $\mathcal{L}_{\Gamma}(\mathcal{BP})$ the set of configurations from which $\mathcal{BP}$ has an accepting path and by $\mathcal{L}(\mathcal{BP})$ the set of words accepted by $\mathcal{BP}$. 
Pushdown Automata and Visibly Pushdown Automata, which we abbreviate by PDA and VPA respectively, are defined analogously, with the only difference being that we only consider finite runs over finite input words and a word is accepted if an accepting state is reached after the last input symbol. We call an automaton \textit{deterministic} (signified by prepending a \textit{D} to the model class) if there is a unique transition for every input symbol, top of stack symbol and control location. We also consider deterministic and non-deterministic finite automata (DFA and NFA) which we define as PDA with only internal transitions.
For all these models, the size of the automaton, denoted $|\mathcal{A}|$ for an automaton $\mathcal{A}$, is the sum of its number of states and number of transitions. 
For PDSs, we will make use of the following Theorem:
\begin{theorem}[\cite{Schwoon2002}]
The emptiness of $\mathcal{L}_{\Gamma}(\mathcal{P})$ for a BPDS $\mathcal{P}$ can be checked in time polynomial in the size of $\mathcal{P}$. Indeed, there is an NFA $\mathcal{A}$ with $\mathcal{L}(\mathcal{A}) = \mathcal{L}_{\Gamma}(\mathcal{P})$ with size polynomial in the size of $\mathcal{P}$.
\end{theorem}
We call the languages accepted by DFA \textit{regular languages} (REG) and those accepted by VPA \textit{visibly pushdown languages} (VPL). The languages accepted by the corresponding models working on infinite words are called $\omega$-regular and $\omega$-visibly pushdown languages ($\omega$-REG and $\omega$-VPL respectively).
Following \cite{Axelsson2010}, we call a class $U$ of finite word automata \textit{reasonable} iff it contains automata recognising the languages $\Sigma$ and $\Sigma^*$
and for any $\mathcal{A} \in U$ and $\mathcal{B}$ of the same automaton type with $\mathcal{L}(\mathcal{A}) = \mathcal{L}(\mathcal{B})$, $\mathcal{B} \in U$ (closure under \textit{equivalences}). All types of finite word automata in this paper are reasonable classes of automata.

Let $\textit{DIR} = \{\downarrow, \downarrow_a\}$. An Alternating Jump Automaton (AJA) \cite{Bozzelli2007} is a $5$-Tuple $(Q, \Sigma, \delta, q_0, \Omega)$ where $Q$ is a finite set of states, $\Sigma$ is a finite pushdown alphabet, $\delta: Q \times \Sigma \rightarrow \mathbb{B}^+(\text{\textit{DIR}} \times Q \times Q)$ is a transition function, $q_0 \in Q$ is an initial state, and $\Omega: Q \rightarrow \mathbb{N}$ is a colouring function. Here, $\mathbb{B}^+(X)$ denotes the set of positive boolean formulae over $X$.\footnote{Note that the model of AJA introduced in \cite{Bozzelli2007} allows further directions that are not needed in this paper.}

Let $\mathcal{A}  =(Q, \Sigma, \delta, q_0, \Omega)$ be an AJA and $w=\alpha_1\alpha_2 \dots \in \Sigma^\omega$ an infinite word. In order to define executions of AJA, we need two types of successor relations of indices on $w$. The direct successor is simply $succ(\downarrow, i) = i+1$. The \emph{abstract} successor $succ(\downarrow_a, i)$ is the index of the input symbol to be read on the same recursion level for a call or internal action, if it exists. In all other cases, it is $\top$. Formally, $succ(\downarrow_a, i) = min \{j > i \mid |v|_{\Sigma_{call}} = |v|_{\Sigma_{ret}} \text{ for } v = \alpha_i \dots \alpha_{j-1}\}$ if it exists and $\alpha_i \in \Sigma_{call} \cup \Sigma_{int}$, where $|v|_X$ is the number of occurences of symbols from $X$ in $w$. Otherwise, $succ(\downarrow_a, i) = \top$.
The AJA $\mathcal{A}$ processes the input word $w$ by reading it from left to right starting in its initial state $q_0$. 
Then, whenever $\mathcal{A}$ is in a state $q$ and reads a symbol $\alpha_i$, it guesses a set of targets $T \subseteq (\text{\textit{DIR}} \times Q \times Q)$ which satisfies $\delta(q,\alpha_i)$. Afterwards, it creates a copy of itself for each $(d, q', q'') \in T$, which then moves to the state $q'$ if $j \! = \! succ(d, \alpha_i) \neq \top$ and reads the symbol $\alpha_j$ next or otherwise transitions into state $q''$ and processes the next symbol $\alpha_{i+1}$. 
Formally, an execution of $\mathcal{A}$ on $w$ is an infinite tree $T(\mathcal{A}, w)$ with the root node $\epsilon$ and the following properties:
1) each node of $T(\mathcal{A},w)$ is associated with a pair $(i,q) \in \mathbb{N} \times Q$ which indicates that a copy of $\mathcal{A}$ is currently in state $q$ and reads the input symbol $\alpha_i$ next. For a node $v$ the associated pair is denoted by $p(v)$ and $p(\epsilon) = (1, q_0)$.
2) For each node $v$ within the tree with $p(v) = (i, q)$ there exists a set of targets $T = \{(d_1,q_1,q'_1) \dots (d_k,q_k,q'_k)\} \subseteq \textit{DIR} \times Q \times Q$  that satisfies $\delta(q,\alpha_i)$ such that for each child node $v_h$ of $v$ with $1 \leq h \leq k$ we have $p(v_h) = (i+1, q'_h)$ if $succ(d_h,w_i) = \top$ and $p(v_h) = (succ(d_h,w_i), q_h)$ otherwise.
A branch $\beta \! = \! v_0v_1...$ of $T(\mathcal{A},w)$ is an infinite sequence of nodes with $v_0 = \epsilon$ and where $v_i$ is the parent of $v_{i+1}$ for each $i \in \mathbb{N}_0$. The set of colours which appear infinitely often throughout a branch is defined as:
$$\mathcal{C}(\beta) = \big\{ \Omega(q) \; \mid \text{ for infinitely many } j \in \mathbb{N}_0 \; \text{, there is }i \in \mathbb{N}_0 \text{ with }p(v_j) = (i,q)\}$$
The execution tree $T(\mathcal{A},w)$ is accepting iff for each branch $\beta$ the smallest colour in $\mathcal{C}(\beta)$ is even (parity condition).
\begin{theorem}[\cite{Bozzelli2007}]\label{theorem:AJAemptiness}
For every AJA $\mathcal{A}$, one can construct in exponential time a BVPS $\mathcal{A}'$ with size exponential in the size of $\mathcal{A}$ such that $\mathcal{L}(A) = \mathcal{L}(\mathcal{A}')$.
\end{theorem}
For simplicity, we define an Alternating Büchi Automaton (ABA) to be an AJA with parity $2$ in which only internal transitions are used (i.e. all input symbols belong to $\Sigma_{int}$). We likewise define a non-deterministic Büchi Automaton (BA) to be an ABA which uses only disjunctions, but no conjunctions. 
In ABA, we allow $\varepsilon$-transitions which can trivially be eliminated.
By the well-known Miyano-Hayashi construction, we obtain the following Theorem:
\begin{theorem}[\cite{Miyano1984}]
The emptiness problem for ABA is PSPACE-complete. Every ABA $\mathcal{A}$ can be converted to a BA $\mathcal{A}'$ with $\mathcal{L}(\mathcal{A}) = \mathcal{L}(\mathcal{A})'$ such that the size of $\mathcal{A}'$is exponential in the size of $\mathcal{A}$. 
\end{theorem}

	\section{Logics}\label{section:logics}
\subsection{Extended LTL}
Let $U$ be a reasonable class of automata. For a fixed set $AP$ of atomic propositions, an LTL[U] formula is given by the following grammar:
$$\varphi ::= ap \mid \neg \varphi \mid \varphi \land \varphi \mid \varphi\;\mathcal{U}^{\mathcal{A}}\;\varphi.$$
The semantics of an LTL[U] formula for a path $\pi$ is given as follows:
\begin{itemize}
\item $\pi \models ap$ iff $ap \in \lambda(\pi(0))$
\item $\pi \models \neg \varphi$ iff $\pi \not\models \varphi$
\item $\pi \models \varphi_1 \land \varphi_2$ iff $\pi \models \varphi_1$ and $\pi \models \varphi_2$
\item $\pi \models \varphi_1 \mathcal{U}^{\mathcal{A}}\varphi_2$ iff $\exists k: \pi^k \models \varphi_2$ and $\forall j < k: \pi^j \models \varphi_1$ and $Actions(\pi, k) \in \mathcal{L}(\mathcal{A})$
\end{itemize}
We denote by $\mathcal{L}(\varphi)$ the set of traces of paths fulfilling $\varphi$. For a KTS $\mathcal{K}$, we define $\mathcal{K} \models \varphi$ iff $Traces(\mathcal{K}) \subseteq \mathcal{L}(\varphi)$. For LTL[U] (and later logics), we denote by $|\varphi|$ the sum of the number of operators in the syntax tree of $\varphi$ and the size of the automata $\mathcal{A}$ occuring in $\varphi$. For simplicity, we often write $\mathcal{U}^{\mathcal{L}}$ for some language $\mathcal{L}$ and use this as an abbreviation for $\mathcal{U}^{\mathcal{A}}$ where $\mathcal{A}$ is the minimal automaton with $\mathcal{L}(\mathcal{A}) = \mathcal{L}$. We also write LTL[REG] and LTL[VPL] to emphasise that it is not relevant whether a deterministic or non-deterministic automaton is used for the respective language class, while we use the specific class of automata (e.g. LTL[DVPA] or LTL[NFA]) otherwise. If no parameter is given for a modality, the language $\Sigma^*$ is assumed.
We can define $\lor$, $\Rightarrow$, $true$ and $false$ in the obvious manner. We define $\varphi_1 \mathcal{R}^{\mathcal{A}} \varphi_2 \equiv \neg (\neg \varphi_1 \mathcal{U}^{\mathcal{A}} \neg \varphi_2)$, $\mathcal{X} \varphi \equiv true\;\mathcal{U}^{\Sigma} \varphi$, $\mathcal{G}^{\mathcal{A}} \varphi \equiv false \mathcal{R}^{\mathcal{A}} \varphi$ and $\mathcal{F}^{\mathcal{A}} \varphi \equiv true\;\mathcal{U}^{\mathcal{A}}\varphi$ . The explicit semantics of $\mathcal{R}^{\mathcal{A}}$ is as follows:
$\pi \models \varphi_1 \mathcal{R}^{\mathcal{A}} \varphi_2$ iff $\forall i: Actions(\pi, i) \notin \mathcal{L}(\mathcal{A})$ or $\pi^i \models \varphi_2$ or $\exists j < i: \pi^j \models \varphi_1$.

We define an LTL[U] formula to be in \textit{Negation Normal Form} (NNF) if negations occur only in front of atomic propositions. 
In order to convert LTL[U] formulae to NNF, we assume $\mathcal{R}^{\mathcal{A}}$ to be a first class modality and $\lor$ to be a first class operator of LTL[U].
 The equivalence $\neg (\varphi_1\;\mathcal{U}^A\;\varphi_2 ) \equiv (\neg \varphi_1)\;\mathcal{R}^{\mathcal{A}} (\neg \varphi_2) $ combined with the classical equivalences for LTL can be used to obtain:
\begin{theorem}\label{theorem:LTLNNF}
Every LTL[U] formula can be converted into an LTL[U] formula in NNF for every $U$.
\end{theorem}
In LTL[U], we can express the example properties mentioned in the introduction in a straightforward manner after choosing an appropriate language class and automaton. Given a DFA $\mathcal{A}$ accepting all finite words of even length, the property \textit{Every other index fulfills $\varphi$} can then be expressed by the LTL[DFA] formula $\mathcal{G}^{\mathcal{A}}\varphi$. Likewise, for a VPA $\mathcal{A}$ recognising the language $req^n grant^n$ (if $req$ is marked a call and $grant$ a return symbol), the LTL[VPA] formula $work\;\mathcal{U}^{\mathcal{A}}\;complete$ states that an agent takes $n$ requests and then grants the same number before completing its work.
Unlike Extended CTL, LTL[U] can also express the classical fairness property $\mathcal{G} \mathcal{F} q$ which, using additional automata, can also be refined by requirements that certain prefixes of a fair path must be contained in a language. 

	\subsection{Extended CTL$^{*}$ and CTL$^{+}$}
Let $U$ be a reasonable class of automata. A CTL$^*[U]$ formula $\varphi$ is given by the following grammar:
$$\varphi ::= ap \mid \neg \varphi \mid \varphi \land \varphi \mid \varphi\;\mathcal{U}^{\mathcal{A}} \; \varphi \mid E \varphi$$
A CTL$^+[U]$ formula $\psi$ is given by the following grammar:
\begin{align*}
    \varphi &::= \psi\;\mathcal{U}^{\mathcal{A}}\; \psi \mid \neg \varphi \mid \varphi \land \varphi\\
    \psi &::= ap  \mid \neg \psi \mid \psi \land \psi \mid  \; E \varphi 
\end{align*}
The semantics of all operators except existential quantification is inherited from LTL[U] in the obvious manner. 
For the quantifier, we define: $\pi \models \exists \varphi \text{ iff } \exists \pi' \in Paths(\pi(0)): \pi' \models \varphi$.
We use the same abbreviations as for LTL[U] and additionally define $A \varphi \equiv \neg E \neg \varphi$. 
To compare our logics with the logics in \cite{Axelsson2010}, we define CTL[$V$,$ W$] to be the sub-logic of CTL$^*$[$V\cup W$] in which all modalities are quantified, quantifiers appear only in front of modalities and additionally, in $\mathcal{U}^{\mathcal{A}}$, $\mathcal{A} \in V$ and in $\mathcal{R}^{\mathcal{A}}$, $\mathcal{A} \in W$. 

Extended CTL$^*$ allows us to express all properties expressible in Extended LTL and Extended CTL. In particular, we can combine linear time and branching time properties, e.g. to refine Extended CTL properties by fairness which itself is not expressible in Extended CTL.

	\section{Automata for $LTL[U]$}\label{section:automataLTLU}
\subsection{LTL[\textit{REG}] formulae to Büchi automata} \label{sec:ltlbuchi}
In this section, we show how LTL[REG] formulae can be translated to ABA. Our translation is inspired by the classical translation of LTL into ABA \cite{Vardi1997} which has also been refined for other extensions with regular modalities \cite{Faymonville2014}.
\begin{theorem}\label{theorem:LTLREGAutomata}
	For every LTL[\textit{REG}] formula $\varphi$ in NNF, there is an ABA $\mathcal{A}$ with $\mathcal{L}(\mathcal{A}) = \mathcal{L}(\varphi)$ with size linear in $|\varphi|$.
\end{theorem}
	\begin{proof}
		Let $\varphi$ be an LTL[REG] formula.  The proof is by structural induction over $\varphi$ with the induction hypothesis that for every subformula $\psi$ of $\varphi$, there is an ABA $\mathcal{A}$ with $\mathcal{L}(\mathcal{A}) = \mathcal{L}(\psi)$ and $|\mathcal{A}| \in \mathcal{O}(|\psi|)$.
		\begin{description}
		\item[Case $\varphi = ap \in AP$.]
		Clearly, the automaton $\mathcal{A} = (\{q_0\}, 2^{AP} \times \Sigma, \delta, q_0, \{q_0\})$ with $\delta(q_0, (X, \sigma)) = true$ for all $\sigma \in \Sigma$ and $X \subseteq AP$ with $ap \in X$ fulfills the claim.
		\item[Case $\varphi = \neg ap$.]
		Similar to the previous case.
		\item[Case $\varphi = \psi_1 \lor \psi_2$.] 
		Let $\mathcal{A}_i = (Q^i, 2^{AP} \! \times \Sigma, \delta^i, q_0^i, F^i)$ be the ABA with $\mathcal{L}(\mathcal{A}_i) = \mathcal{L}(\psi_i)$ for $i \in \{1,2\}$ given by the induction hypothesis. W.l.o.g. let $Q^1 \cap Q^2 = \emptyset$. The idea in this case is to shift the semantics for the $\lor$ operator directly into the transition function.
		The ABA $\mathcal{A} = (Q, 2^{AP} \!\times \Sigma, \delta, q_0, F)$ accepting the language $\mathcal{L}(\varphi)$ consists of a state set $Q =  Q^1 \cup Q^2 \cup \{q_0\}$ and the set of accepting states $F= F^1 \cup F^2$. It starts from a new initial state $q_0$ from which it is able to transition into both automata $\mathcal{A}_1$ and $\mathcal{A}_2$ using a disjunctive $\varepsilon$-transition.

		\item[Case $\varphi = \psi_1 \land \psi_2$.]
		Similar to the previous case.
		\item[Case $\varphi = \psi_1 \mathcal{U}^\mathcal{X} \psi_2$.]
		Again, the induction hypothesis yields the ABA $\mathcal{A}_i = (Q^i, 2^{AP} \!\times \Sigma, \delta^i, q_0^i, F^i)$ with $\mathcal{L}(A_i) = \mathcal{L}(\psi_i)$ for $i \in \{1,2\}$. Furthermore let $\mathcal{X}=(Q',\Sigma, \delta', q_0', F')$ be the NFA used in the formula. It is assumed that $Q^1, Q^2$ and $Q'$ are pairwise disjoint. Intuitively, we translate the formula just like a basic LTL formula, but additionally run the NFA to guess an accepting path, requiring $\varphi_1$ to hold until the accepting state is hit and then requiring $\varphi_2$ to hold. 
		Formally, let $\mathcal{A}_{\varphi} = (Q, \hat{\Sigma}, \delta, F)$ where $Q = Q_1 \cup Q_2 \cup Q'$, $\delta = \delta_1 \cup \delta_2 \cup \hat{\delta'}$, 
		$q_0 = q_0'$ and $F = F_1 \cup F_2$. Finally, we set 
		 and $\hat{\delta'}(q', (X, \sigma)) = \delta_1(q_0^1, (X, \sigma)) \land \bigvee_{q'' \in \delta'(q', \sigma)}q''$ for all $q' \in Q'$ and $\hat{\delta'}(q', \varepsilon) = q_0^2$ for $q' \in F'$.
		\item[Case $\varphi = \psi_1 R^\mathcal{X} \psi_2$.]
		Once more, the ABA $\mathcal{A}_i = (Q^i, 2^{AP}\!\times\Sigma, \delta^i, q_0^i, F^i)$ with $\mathcal{L}(A_i) \! = \! \mathcal{L}(\psi_i)$ for $i \in \{1,2\}$ are provided by the induction hypothesis.  Again let $\mathcal{X}=(Q',\Sigma, \delta', q_0', F')$ be the NFA used in the formula and let $Q^1, Q^2$ as well as $Q'$ be pairwise disjoint. 
Let $\mathcal{A}_{\varphi} = (Q, 2^{AP} \times \Sigma, \delta, F)$ where $Q = Q_1 \cup Q_2 \cup Q'$, $\delta = \delta_1 \cup \delta_2 \cup \hat{\delta'}$, 
		$q_0 = q_0'$ and $F = F_1 \cup F_2 \cup Q'$. Finally, we set 
		$\hat{\delta'}(q', \varepsilon) = q_0^1 \land q_0^2$ and  $\hat{\delta'}(q', (X, \sigma)) = \delta_2(q_0^2, (X, \sigma)) \land \bigwedge_{q'' \in \delta'(q', \sigma)}q''$ for final states $q' \in F$.
For non-final states $q' \notin F$, we set $\hat{\delta'}(q', \varepsilon) = q_0^1$  and $\hat{\delta'}(q', (X, \sigma)) =  (\bigwedge_{q'' \in \delta'(q', \sigma)}q'') $.
For this modality, we have to distinguish several cases: if $\varphi_1 \land \varphi_2$ is fulfilled at a final state, the whole formula is fulfilled there. On the other hand, it suffices to require only $\varphi_1$ to hold if the state currently under consideration in the NFA is non-accepting. If $\varphi_1$ , we require $\varphi_2$ to hold unless the current state is non-accepting, in which case we do not require any formulae to hold. However, unless $\varphi_2$ is released from holding by one of these cases, we always have to pursue all possible successors in the NFA lest we miss paths which reach an accepting state and thus require $\varphi_2$ to hold. All accepting states of the automata $\mathcal{A}_i$ stay accepting in order to reflect the semantics of $\psi_i$. Moreover, all states of the NFA $\mathcal{X}$ are declared accepting as paths on which $\varphi_1$ never holds, but $\varphi_2$ holds on any accepting state of $\mathcal{X}$ are allowed by the semantics of the Release operator.\hfill\qed
	\end{description}	
\end{proof}
	\subsection{Automata for LTL[VPL]}
Just as for LTL[REG], we construct alternating automata for LTL[VPL], this time using AJA instead of ABA. 


\begin{theorem}\label{theorem:VPLautomata}
	Every LTL[\textit{VPL}] formula $\varphi$ can be translated into an AJA $\mathcal{A}$ with $\mathcal{L}(\varphi) \! = \! \mathcal{L}(\mathcal{A})$ and $|\mathcal{A}| \in \mathcal{O}\big( |\varphi|^2 \big)$.
\end{theorem}
	\begin{proof}
		Let $\varphi$ be an LTL[\textit{VPL}] formula over a finite proposition set $AP$ and a set of actions $\Sigma$. The proof is done by structural induction over $\varphi$, with the induction hypothesis that for every subformula $\psi$ of $\varphi$ there is an AJA $\mathcal{A}$ with $\mathcal{L}(\mathcal{A}) = \mathcal{L}(\psi)$ and $|\mathcal{A}| \in \mathcal{O}\big( |\varphi|^2 \big)$. Atomic propositions and logical operators are handled just as in \autoref{theorem:LTLREGAutomata}. A formula $\neg \varphi$ can be handled by complementing the AJA for $\varphi$ which is possible by shifting the parity without any blowup and dualising the transition function. Thus, we do not require NNF here. Notice that this would not have been possible for LTL[REG] in the proof of \autoref{theorem:LTLREGAutomata} since complementing an ABA induces a quadratic blowup and this would altogether lead to an exponential blowup. 
Hence, it only remains to consider formulae of the form $\varphi = \psi_1 \mathcal{U}^{\mathcal{X}} \psi_2$. Here, just as in \cite{Weinert2018}, we can use an AJA to simulate the stack of the VPA $\mathcal{X}$ by introducing a flag in order to indicate whether a call should ever be returned from and enforcing this with the acceptance condition, and following the execution after the return, where needed, via a conjunctive copy. We can otherwise linearly follow the trace and require $\phi_1$ to hold unless we find an accepting state in the VPA $\mathcal{X}$ and $\phi_2$ holds. \hfill\qed
	\end{proof}

	\section{Expressivity}\label{section:expressivity}
In this section, we compare the expressive power of several logics to each other. For two logics $L, L'$, we write $L \leq L'$ if every formula $\varphi$ in $L$ has a matching formula $\varphi'$ in $L'$ such that for all KTSs $\mathcal{K}$, $\mathcal{K} \models \varphi$ iff  $\mathcal{K} \models \varphi'$. This induces the expressivity relation $<$ in the obvious way. Furthermore, we write $L \leq_{lin} L'$ (resp. $L \leq_{exp} L'$ or $L \leq_{poly} L'$) to denote that the translation of formulae in $L$ to formulae in $L'$ involves linear (resp. exponential or polynomial) blowup.
The following inclusions are immediate:
\begin{theorem}
\begin{enumerate}
\item LTL[U] $\leq_{lin}$ CTL$^*$[U]
\item CTL[U, V] $\leq_{lin}$ CTL$^+$[W] if $U \subseteq W$ and $V \subseteq W$
\item CTL$^+$[U] $\leq_{lin}$ CTL$^*$[U]
\end{enumerate}
\end{theorem}
Furthermore, it was shown in \cite{Axelsson2010} that CTL[U, V] cannot express the basic fairness property $\mathcal{G}\mathcal{F}q$ regardless of the automata classes chosen. This immediately implies:
\begin{theorem}
CTL$^*$[W] $\not \leq$ CTL[U, V] for arbitrary classes $U, V, W$.
\end{theorem}
In order to discuss some model theoretic properties needed for some of the embeddings of our logics, we introduce the logic PDL-$\Delta$[U] \cite{Axelsson2010}.
For a set $\Pi$ of atomic programs, PDL-$\Delta$[U] is defined by the following rules:
\begin{itemize}
\item Every atomic proposition $ap$ is a formula. 
\item If $\varphi_1$ and $\varphi_2$ are formulae, then so are $\varphi_1 \land \varphi_2$ and $\neg \varphi_1$. 
\item If $\varphi$ is a formula, then $\varphi?$ is a test where the set of tests is denoted $Test$.
\item A regular expression over $\Pi \cup Test$ is a program. 
\item If $\alpha$ is a program and $\varphi$ is a formula, then $\langle \alpha \rangle \varphi$ is a formula.
\item An automaton $\mathcal{A}$ of type $U$ over the alphabet $\Sigma \cup Test$ is an $\omega$-program.
\item For every $\omega$-program $\mathcal{A}$, $\Delta \mathcal{A}$ is a formula. 
\end{itemize}
A formula of PDL-$\Delta$[U] is evaluated over a structure $\mathcal{M} = (S, R, v)$ where $S$ is a set of states, $R : \Pi \rightarrow 2^{S \times S}$ is a transition relation that assigns to atomic programs the state transitions that are possible using them and $v: S \rightarrow 2^{AP}$ assigns atomic propositions to states. We can interpret $\mathcal{M}$ as a KTS by interpreting atomic programs as actions.  
The relation $R$ is extended to tests as follows:
$R(\varphi?) = \{(s, s) \mid M, s \models \varphi\}$. Moreover, for programs $\alpha$,  $R(\alpha) = \{(s, s') \mid \exists w = w_1 \dots w_m \in \mathcal{L}(\alpha): \exists s_0 \dots s_m \in S: s = s_0 \land s' = s_m \land (s_{i-1}, s_i) \in R(w_i)$ for all $1  \leq i \leq m \}$.
Furthermore, we have for each automaton $\mathcal{A}$ over the alphabet $\Pi \cup Test$ a unary relation $R_{\omega}(\mathcal{A})$ such that $s \in R_{\omega}(\mathcal{A})$ iff there is an infinite word $w = w_0 w_1 \dots \in \mathcal{L}(\mathcal{A})$ and a sequence of states $s_0 s_1 \dots$ such that $s_0 = s$ and $(s_i, s_{i+1}) \in R(w_i)$ for all $i \geq 0$. We note that if $\mathcal{A}$ is a VPS, $Test$ is assumed to belong to the internal symbols.

The semantics of PDL-$\Delta$[U] for a structure $\mathcal{M}$ and a state $s \in S$ is then given as follows:
\begin{enumerate}
\item $\mathcal{M}, s \models ap$ iff $ap \in v(s)$
\item $\mathcal{M}, s \models \varphi_1 \land \varphi_2$ iff $\mathcal{M}, s \models \varphi_1$ and $\mathcal{M}, s \models \varphi_2$
\item $\mathcal{M}, s \models \neg \varphi$ iff $\mathcal{M}, s \not\models \varphi$
\item $\mathcal{M}, s \models \langle \alpha \rangle \varphi$ iff there is $s' \in S$ such that $(s, s') \in R(\alpha)$ and $s' \models \varphi$
\item $\mathcal{M}, s \models \Delta \mathcal{A}$ iff $s \in R_{\omega}(\mathcal{A})$
\end{enumerate}

\begin{theorem}\label{theorem:CTL*embedding}
Let $c \in \{*, +\}$.
\begin{enumerate}
\item CTL$^c$[REG] $<_{exp}$ PDL-$\Delta$[REG].
\item CTL$^c$[VPL] $<_{exp}$ PDL-$\Delta$[VPL].
\end{enumerate}
\end{theorem}
\begin{proof}
For the innermost existential formula $E \psi$ of a CTL$^c$[REG] formula $\varphi$, we build the corresponding ABA $\mathcal{A}_{\psi}$ and dealternate it using the Miyano-Hayashi construction to obtain an equivalent BA $\mathcal{A}'$ with exponentially more states. We can thus replace $E \psi$ by the formula $\psi' \equiv  \Delta \mathcal{A}'$. Inductively, we always temporarily replace such a formula by a fresh atomic proposition and then, after the dealternation, replace this proposition by the test $(\psi' ?)$ in the new automaton. This way, we can integrate the BAs into each other. This translation is exponential due to the dealternation. 
For $CTL^c[VPL]$, the proof is analogous, replacing ABA by AJA (and using \autoref{theorem:AJAemptiness} for the dealternation). \hfill\qed
\end{proof}

For PDL-$\Delta[U]$, we make use of the following model-theoretic results:
\begin{theorem}[\cite{Axelsson2010}]
\begin{enumerate}
\item PDL$-\Delta$[REG] has the finite model property.
\item PDL$-\Delta$[VPL] has the \textit{visibly pushdown model property}, i.e. every satisfiable formula has a model that is a VPS.
\end{enumerate}
\end{theorem}
Due to the embedding of \autoref{theorem:CTL*embedding} and the fact that there is already a satisfiable CTL[VPA, VPA] formula with no finite model \cite{Axelsson2010}, we obtain:
\begin{corollary}
Let $c \in \{*, +\}$.
\begin{enumerate}
\item CTL$^c$[REG] has the finite model property. 
\item There is a satisfiable CTL$^c$[VPL] formula with no finite model. CTL$^c$[VPL] has the visibly pushdown model property.
\end{enumerate}
\end{corollary}
\begin{corollary}
Let $c \in \{*, +\}$.  CTL$^c$[REG] $<$   CTL$^c$[VPL] 
\end{corollary}

\begin{theorem}
LTL $<$ LTL[REG] $<$ LTL[VPL] $<$ LTL[DPDA].
\end{theorem}
\begin{proof}
In every case, the relation $\leq$ is clear. LTL $<$ LTL[REG] follows from Wolper's classical argument \cite{Wolper1983} that LTL cannot express that a proposition occurs on every other index. For every formula $\varphi$ of LTL[REG], there is an ABA recognising $\mathcal{L}(\varphi)$. Thus, $\mathcal{L}(\varphi)$ is $\omega$-regular.
On the other hand, the language of $\mathcal{F}^{a^n b^n} (\mathcal{G}^{\overline{\{c\}}*}\;false)$
is not $\omega$-regular.
%
The last strict inclusion follows from the fact that LTL[DPDA] satisfiability checking is undecidable, but the same problem is decidable for the other fragments (see \autoref{section:satisfiability}). \hfill\qed
\end{proof}
By considering the class of linear KTS (which only have one path), we get:
\begin{corollary}
CTL$^c$[VPL] $<$  CTL$^c$[DPDA] for $c \in \{*,+\}$.
\end{corollary}

\begin{theorem}
LTL[U] and CTL[V,W] have incomparable expressivity for arbitrary classes $U$, $V$ and $W$ 
\end{theorem}
\begin{proof}
Since LTL[U] formulae are invariant under trace equivalence, but CTL[V, W] formulae are not (regardless of the choice of $V$ and $W$), there can be no embedding of CTL[V, W] into LTL[U].
On the other hand, if $LTL[U] < CTL[V, W]$ held for some classes $U, V, W$, then the LTL formula $\varphi \equiv \mathcal{F} \mathcal{G}\neg q$ would be expressible in $CTL[V]$. 
Note that this formula is equivalent to $A \varphi$ on structures. Since $CTL[V]$ is closed under negation, the formula $\psi \equiv E \mathcal{G} \mathcal{F} q$ would then also be expressible in $CTL[V, W]$. This contradicts Lemma 4.3 of \cite{Axelsson2010}. \hfill \qed
\end{proof}

	\section{Satisfiability}\label{section:satisfiability}
In this section, we tackle the satisfiability problem for our logics. For this purpose, we call a formula $\varphi$ of a temporal logic \textit{satisfiable} if there is a KTS $\mathcal{K}$ such that $\mathcal{K} \models \varphi$. 
We obtain exhaustive decidability and complexity classifications in all cases.
Unfortunately, unlike for Extended CTL, satisfiability is undecidable for all our logics for language classes going beyond VPL.
\begin{theorem}\label{theorem:LTLDPDAsatisfiability}
LTL[DPDA] satisfiability checking is undecidable.
\end{theorem}
\begin{proof}
For two DCFL $\mathcal{L}_1, \mathcal{L}_2$, let $\hat{\mathcal{L}_i} = \{w \hat{t} \mid w \in \mathcal{L}_i\} $ where $\hat{t}$ is a fresh symbol.
The formula $\mathcal{F}^{\hat{\mathcal{L}_1}}true \land \mathcal{F}^{\hat{\mathcal{L}_2}} true$ is satisfiable iff $\mathcal{L}_1 \cap \mathcal{L}_2$ is non-empty. \hfill\qed
\end{proof}

\begin{corollary}
CTL$^c$[DPDA] satisfiability checking is undecidable for $c \in \{*, +\}$.
\end{corollary}

\noindent
For regular languages, the complexity of the satisfiability problem for LTL[REG] does not increase beyond the complexity of the satisfiability problem for LTL.
 
\begin{theorem}
LTL[REG] satisfiability checking is PSPACE-complete.
\end{theorem}

\begin{proof}
Membership in PSPACE can be shown by applying the emptiness test to the ABA for a formula $\varphi$. Hardness follows from the corresponding hardness of LTL satisfiability checking. \qed
\end{proof}

\noindent
For visibly pushdown languages, the problem remains decidable but the complexity increases because we can express hard problems about VPAs.

\begin{theorem}
LTL[VPL] satisfiability checking is EXPTIME-complete.
\end{theorem}
\begin{proof}

\noindent
Membership in EXPTIME follows from the fact that we can build an AJA characterising $\mathcal{L}(\phi)$ for an LTL[VPL] formula $\phi$ and test it for emptiness in EXPTIME.
For hardness, we reduce from the LTL[VPL] model checking problem for a formula $\phi$ against a finite KTS $\mathcal{K} = (S, \rightarrow, \lambda)$ which is shown to be EXPTIME-complete later in this paper.
For any $s \in S$, we introduce a fresh atomic proposition $a_s$ and write $\psi_s \equiv \bigvee_{(s, t, s') \in \delta} (\mathcal{X}^{\{t\}} (a_{s'} \land \bigwedge_{a \in \lambda(s')}a \land \bigwedge_{a' \notin \lambda(s')}\neg a' ))$. Then $\varphi' = a_{s_0} \land \bigwedge_{a \in \lambda(s_0)}a \land \bigwedge_{a' \notin \lambda(s_0)}\neg a' \land \mathcal{G} (\bigwedge_{s \in S} (a_s \Rightarrow \psi_s)) \land \neg \varphi$ is satisfiable iff there is a path in $\mathcal{K}$ violating $\varphi$.  \qed
\end{proof}
For the satisfiability problem of branching time logics, we follow the approach of \cite{Axelsson2010} and use established results for the satisfiability problem of PDL-$\Delta$[U].
\begin{theorem}
\begin{enumerate}
\item PDL$-\Delta$[REG] satisfiability is EXPTIME-complete.
\item PDL$-\Delta$[VPL] satisfiability is 2EXPTIME-complete.
\end{enumerate}
\end{theorem}
Using our embeddings into PDL-$\Delta$[U], we then obtain the following result:
\begin{theorem}
Let $c \in \{*, +\}$.
\begin{enumerate}
\item CTL$^c[REG]$ satisfiability checking is 2EXPTIME-complete.
\item CTL$^c[VPL]$ satisfiability checking is 3EXPTIME-complete.
\end{enumerate}
\end{theorem}
\begin{proof}
By \autoref{theorem:CTL*embedding}, we have an exponential translation of CTL$^c$[U] into PDL-$\Delta$[U] for $U \in \{\text{REG, VPL}\}$. Furthermore, PDL-$\Delta$[REG] satisfiability checking is EXPTIME-complete, implying a 2EXPTIME upper bound for CTL$^c$[REG]. The lower bound already holds for ordinary CTL$^+$\cite{Johannsen2003}. Since the satisfiability test for PDL$-\Delta[VPL]$ is possible in 2EXPTIME, we obtain a 3EXPTIME upper bound and the lower bound follows from the corresponding lower bound for CTL[DVPA,DVPA$\cup$ NFA] \cite{Axelsson2010}. \qed
\end{proof}
The last theorem shows that parameterisation by regular languages does
not increase the complexity of the satisfiability problem for either
CTL$^*$ or CTL$^+$. While the use of VPL increases the complexity exponentially, this already holds for Extended CTL. From this point of view, the additional expressive power of our logics is obtained for free in this case as well.
Finally, the lower bound for CTL$^c$ satisfiability and the complexity of PDL-$\Delta$[U] satisfability imply that there is no polynomial translation of  CTL$^c[U]$ into that logic. Since there is a linear translation of CTL[U,V] into PDL-$\Delta$[$U \cup V]$, a possible translation of Extended CTL$^+$ into Extended CTL (if any) must involve an exponential blowup:
\begin{corollary}
\begin{enumerate}
\item CTL$^c$[U] $\not\leq_{poly}$ PDL$-\Delta$[U] for $U \in \{\text{REG, VPL}\}$.
\item CTL$^+$[U] $\not\leq_{poly}$ CTL[U, U] for $U \in \{\text{REG, VPL}\}$.
\end{enumerate}
\end{corollary}

	\section{Model Checking}\label{section:modelchecking}


\subsection{Finite Kripke Models}
We begin with the discussion of DPDAs. Just as the satisfiability problem, the model checking problem is not decidable for this class:
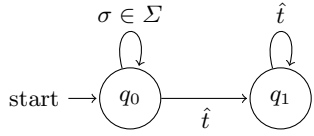
\begin{figure}[t]
\begin{center}
\begin{tikzpicture}[shorten >=1pt,node distance=2cm,on grid,auto] 
   \node[state,initial] (q_0)   {$q_0$}; 
   \node[state] (q_1) [right=of q_0] {$q_1$}; 
    \path[->] 
    (q_0) edge [loop above] node {$\sigma \in \Sigma$} (q_0)
          edge  node [swap] {$\hat{t}$} (q_1)
    (q_1) edge [loop above] node  {$\hat{t}$} (q_1);
\end{tikzpicture}
\vspace{-5mm}
\end{center}
\caption{Structure generating words with arbitrary prefixes}
\label{fig:wordgen}
\end{figure}

\begin{theorem}\label{theorem:CTLDPDA}
LTL[DPDA] model checking against KTS is undecidable.
\end{theorem}
\begin{proof}
The formula $\neg (\mathcal{F}^{\mathcal{L}_1\;\hat{t}} end \land \mathcal{F}^{\mathcal{L}_2\;\hat{t}} true)$ is fulfilled by the structure in \autoref{fig:wordgen} iff $\mathcal{L}_1 \cap \mathcal{L}_2$ is empty for two deterministic context-free languages $\mathcal{L}_1, \mathcal{L}_2$.\hfill\qed
\end{proof}
Since the argument in the last proof can be reused for the other logic, we obtain:

\begin{corollary}\label{theorem:CTLDPDA}
CTL$^c$[DPDA] model checking against finite KTS is undecidable for $c \in \{*, +\}$.
\end{corollary}
As a direct consequence, there can be no embedding of CTL$^+$ equipped with DPDA into corresponding CTL variants, which is surprising because the basic logics CTL$^+$ and CTL have the same expressive power \cite{Emerson1982}:
\begin{corollary}
 CTL$^+$[DPDA] $\not <$  CTL[DPDA,DPDA].
\end{corollary}
\begin{proof}
CTL$^+$[DPDA] model checking against a Kripke model is undecidable, but CTL[DPDA, DPDA] model checking is P-complete \cite{Axelsson2010}. \qed
\end{proof}
The last result shows that there is no generic translation from Extended CTL$^+$ into Extended CTL that is uniform for all types of automata.

For regular languages, the complexity of model checking is not increased beyond the complexity of standard LTL for Extended LTL.
\begin{theorem}\label{theorem:LTLREGKripke}
LTL[REG] model checking against finite KTS is PSPACE-complete and in P for a fixed formula.
\end{theorem}
\begin{proof}
We can check in PSPACE if the language of a KTS is a subset of the language of the automaton constructed in \autoref{theorem:LTLREGAutomata}. PSPACE-hardness follows trivially from the PSPACE-hardness of LTL model checking.
If the formula is fixed, the size of the ABA is constant as well as the size of the BA obtained by the Miyano-Hayashi construction. The emptiness test after synchronisation with the KTS is possible in polynomial time. \hfill\qed
\end{proof}

\noindent
As for satisfiability, switching to VPLs increases the complexity to EXPTIME:

\begin{theorem}\label{theorem:LTLVPLKripke}
LTL[VPL] model checking against finite KTS is EXPTIME-complete and in $P$ for a fixed formula.
\end{theorem}

\begin{proof}
We reduce the problem of model checking an LTL formula $\varphi$ against a VPS $\mathcal{A} = (Q, \Gamma, \Delta, \lambda)$ which is EXPTIME-complete \cite{Schwoon2002}\footnote{In \cite{Schwoon2002}, the lower bound was established for LTL model checking against a PDS, but this can easily be adopted by transforming the PDS and the formula in the standard way.}.
For this, we define a KTS $\mathcal{K}_{\mathcal{A}} = (Q \times \Gamma, \rightarrow, \lambda)$ as a particular regular overapproximation of $\mathcal{A}$. The KTS $\mathcal{K}_{\mathcal{A}}$ follows the evolution of the 
configuration heads $(q,\gamma)$ in evolutions of $\mathcal{A}$. It can do so precisely for call and internal steps but guesses the topmost stack symbol in the configuration reached after a return step. In order to recover the actual executions of $\mathcal{A}$ an adequately defined DVPA $\mathcal{A'}$ is used in the formula. In order to allow $\mathcal{A'}$ to do so, the KTS $\mathcal{K}_{\mathcal{A}}$ makes visible in the actions the symbol pushed onto the stack for a call rule in a corresponding call-symbol and the stack symbol guessed as target of a return step in a corresponding return-symbol.
Then, for $\psi_{\mathcal{A}'} \equiv \mathcal{G}^{\overline{\mathcal{A}'}} false$ and $\varphi' \equiv \psi_{\mathcal{A}'} \Rightarrow \varphi$, we have $\mathcal{K}_{\mathcal{A}} \models\varphi'$ iff $\mathcal{A} \models \varphi$ since $\psi_{\mathcal{A}'}$ holds iff a path of $\mathcal{K}_{\mathcal{A}}$ corresponds to a proper path of $\mathcal{A}$.

The upper bound follows from \autoref{theorem:AJAemptiness} since we can build an AJA for $\neg \varphi'$ and test it against $\mathcal{K}_{\mathcal{A}}$ (\autoref{theorem:AJAemptiness}). The complexity for the fixed formula follows analogously to the case of LTL[REG].\qed
\end{proof}
Using our results for LTL[REG], we can derive a model checking algorithm for CTL$^c$[REG]:
\begin{theorem}
For $c \in \{*, +\}$, CTL$^c$[REG] model checking against finite KTS is PSPACE-complete. For a fixed formula, it is in $P$.
\end{theorem}
\begin{proof} For the PSPACE-hardness proof, recall that a classical result of complexity theory states that the problem of deciding for $n$ different DFA whether their intersection is non-empty is PSPACE-complete \cite{Kozen1977}.
This problem can easily be reduced to CTL$^c$[DFA] model checking in polynomial time: for DFA $\mathcal{A}_1 \dots \mathcal{A}_n$, the formula $E(\mathcal{F}^{\mathcal{A}_1} true \land \dots \mathcal{F}^{\mathcal{A}_n}true)$ can be checked against the structure given in \autoref{fig:wordgen} to test the intersection under consideration for emptiness.

The model checking proceeds just like the classical CTL$^*$ model checking algorithm \cite{Demri2016}: we pick an innermost existential subformula $E \psi$, replace it by a fresh atomic proposition $p$, check all states of the input KTS with the LTL[REG] model checking algorithm for $\neg \psi$, mark all states fulfilling the latter with $\neg p$ and all others with $p$. Inductively, we obtain a pure LTL[REG] formula which we can check against the KTS. Since every step can be done in polynomial space and we perform at most $|\varphi|$ steps for an input formula $\varphi$, we obtain the PSPACE upper bound. 
The second part follows from \autoref{theorem:LTLREGKripke} since our CTL$^c$[REG] model checking algorithm applies the LTL[REG] algorithm a bounded number of times for fixed formulae. \hfill\qed
\end{proof}
Note that CTL$^*$ model checking against finite KTS is already PSPACE-complete, in contrast to CTL$^+$ model checking which is $\Delta_2^p$-complete \cite{Laroussinie2001}.

Just as the CTL$^c$[REG] model checking algorithm is derived from the LTL[REG] model checking algorithm, the CTL$^c$[VPL] is derived from the LTL[VPL] algorithm:
\begin{corollary}
For $c \in \{*, +\}$, CTL$^c$[VPL] model checking against finite KTS is EXPTIME-complete. For a fixed formula, it is in $P$.
\end{corollary}
\begin{proof}
For the lower bound, we proceed as for LTL[DVPA] in \autoref{theorem:LTLVPLKripke}. The model checking problem for LTL against VPS $\mathcal{A}$ is EXPTIME complete for the fixed formula $\varphi \equiv \mathcal{G}(\neg fin)$ if we introduce \textit{checkpoints}, i.e. DFA accepting configurations of $\mathcal{A}_{(q, \gamma)}$ for every head $(q, \gamma)$ \cite{Schwoon2002}. Each transition for $(q, \gamma)$ is then conditional on the DFA ${A}_{(q, \gamma)}$ accepting the current configuration. This restriction can be modelled by a DVPA representing $\mathcal{A}_{(q, \gamma)}$ and checked by a $\mathcal{G}$-formula as before if we encode the transitions in the transition labels of $\mathcal{K}_{\mathcal{A}}$. We obtain a universally quantified conjunction of $\mathcal{G}^{\mathcal{B}}false$ formulae and $\varphi$ which is a CTL$^+$[DVPA] formula and completes the reduction. The upper bound follows as in the previous theorem, replacing LTL[REG] by LTL[VPL]. \hfill\qed 
\end{proof}

\subsection{Visibly Pushdown Systems}
For VPS, we again make use of our approach based on alternating automata, but lift it by employing classical results of pushdown model checking.
\begin{theorem}\label{theorem:LTLREGVPS}
LTL[REG] model checking against VPS is EXPTIME-complete and in P for a fixed formula.
\end{theorem}
\begin{proof}
Let $\varphi$ be a LTL[REG] formula, $\mathcal{P}$ be a VPS and  $\mathcal{A}$ be the ABA for $\varphi$. From $\mathcal{A}$, we obtain an BA $\mathcal{A}'$ by the Miyano-Hayashi construction with exponential blowup. Synchronizing $\mathcal{A}'$ with $\mathcal{P}$, we obtain a Büchi VPS $\mathcal{P}'$ which we can test for emptiness in polynomial time. All in all, we obtain an EXPTIME-algorithm. The lower bound holds already for LTL. 
If the size of the formula is fixed, the complexity follows with the same argument as before since the only non-bounded complexity term is induced by the polynomial time emptiness test. \hfill\qed
\end{proof}
This time, the complexity is not increased by switching to VPL.
\begin{theorem}
LTL[VPL] model checking against VPS is EXPTIME-complete and in P for a fixed formula. 
\end{theorem}
\begin{proof}
We proceed as in the proof of \autoref{theorem:LTLREGVPS}, but this time, we dealternate an AJA $\mathcal{A}$ for a LTL[VPL] formula $\varphi$ to obtain a Büchi VPS $\mathcal{A}'$ with an exponential blowup. Since VPS can be synchronised, we can again test the synchronised Büchi VPS for emptiness and obtain an EXPTIME algorithm.
The argument for the fixed formula is the same as above. \hfill\qed
\end{proof}
Our AJA for LTL[VPL] can be used to extend the classical CTL$^*$ model checking algorithm for PDS \cite{Schwoon2002} to CTL$^*$[VPL]:
\begin{theorem}\label{theorem:CTL*LTLVPLVPS}
CTL$^*$[VPL] model checking against VPS is 2EXPTIME-complete and EXPTIME-complete for a fixed formula. 
\end{theorem}
\begin{proof}
Hardness follows from 2EXPTIME-hardness of checking a CTL$^*$ formula against a VPS $\mathcal{P}$ \cite{Bozzelli2007a}. 
For inclusion, notice that we can construct an AJA for $\varphi$ of polynomial size by \autoref{theorem:VPLautomata} for the negation of any innermost $LTL[VPL]$ formula preceded by an existential quantifier. For this automaton, we obtain a non-deterministic Büchi VPS recognising the same language. From this VPS, we can construct an NFA $\mathcal{A}_{\psi}$ that accepts all configurations from which the synchronised product of that automaton and $\mathcal{P}$ has an accepting run. Using determinisation with another exponential blowup, we can assume $\mathcal{A}_{\psi}$ to be deterministic and then complement it in linear time to obtain an automaton that accepts configurations fulfilling the existential formula. We can then replace the existential formula $E\psi$ by a fresh atomic proposition $a_{\psi}$ and synchronise $\mathcal{P}$ with the automaton via a regular valuation such that a configuration is labelled $a_{\psi}$ if it fulfills $E \psi$. Inductively, we obtain a pure LTL[VPL] formula and check that formula against $\mathcal{P}$. This algorithm works in 2EXPTIME and thus establishes the desired upper bound. 
If the size of the formula is fixed, the problem is EXPTIME-hard (already for pure CTL \cite{Bozzelli2007a}) and the upper bound follows as the NFA we determinise has bounded size and the determinisation thus produces a DFA of exponential instead of doubly exponential size. \hfill\qed
\end{proof}
Since the lower bounds hold already for pure CTL$^*$, we directly obtain the following result for regular languages:
\begin{corollary}
CTL$^*$[REG] model checking against VPS is 2EXPTIME-complete and EXPTIME-complete for a fixed formula. 
\end{corollary}
Of course, our algorithm can also be applied to CTL$^+$[VPL] formulae. However, we do not obtain completeness in all cases since the exact complexity of CTL$^+$ model checking on VPS (and PDS) - even without language parameters - is an open question.
\begin{corollary}
For $U \in \{\text{DFA, DVPA}\}$, CTL$^+[U]$ model checking against VPS is in 2EXPTIME and  EXPTIME-hard. Moreover, it is EXPTIME-complete for a fixed formula. For $U \in \{\text{NFA, VPA}\}$, it is 2EXPTIME-complete and EXPTIME-complete for a fixed formula.
\end{corollary}
\begin{proof}
The 2EXPTIME upper bound follows from \autoref{theorem:CTL*LTLVPLVPS} and EXPTIME-hardness holds already for CTL model checking against VPS for a fixed formula \cite{Bozzelli2007a}. For a fixed formula, the model checking algorithm of \autoref{theorem:CTL*LTLVPLVPS} works in exponential time because the dealternised AJA has constant size and the determinisation of the NFA only takes exponential time. The 2EXPTIME-hardness of CTL$^+$[NFA] and thus also CTL$^+$[VPA] model checking follows from the 2EXPTIME-hardness of CTL[DFA, NFA] \cite{Axelsson2010}.
\hfill\qed
\end{proof}

\subsection{Pushdown Systems}
For PDS, we can immediately lift the results for VPS when it comes to regular languages since PDS can also trivially be synchronised with BA and therefore we proceed as in the proofs of the corresponding statements for VPS to obtain:
\begin{theorem}
LTL[REG] model checking against PDS is EXPTIME-complete.
\end{theorem}
\begin{theorem}
CTL$^*$[REG] model checking against PDS is 2EXPTIME-complete.
\end{theorem}
\begin{corollary}
CTL$^+$[DFA] model checking against PDS is in 2EXPTIME, EXPTIME-hard and EXPTIME-complete for a fixed formula. CTL$^+$[NFA] model checking is 2EXPTIME-complete and EXPTIME-complete for a fixed formula.
\end{corollary}
However, going beyond regular languages makes the model checking problem undecidable:
\begin{theorem}
LTL[VPL] model checking against PDS is undecidable,
\end{theorem}
\begin{proof}
In \cite{Axelsson2010}, it is shown that model checking a formula of type $\varphi \equiv E\mathcal{F}^{\mathcal{A}} true$ for a particular VPA $\mathcal{A}$ against a PDS is undecidable and $\varphi$ holds iff the LTL[VPL] formula $\mathcal{G}^{\mathcal{A}} false$ does not hold. \hfill\qed
\end{proof}
\begin{corollary}
CTL$^c$[VPL] model checking against PDS is undecidable where $c \in \{*, +\}$.
\end{corollary}

\noindent
Intuitively, the difference between PDS and VPS stems from the fact that VPS can be synchronised with  VPS, while PDS cannot.

	\section{Conclusion and Related Work}\label{section:conclusion}

We introduced extended variants of LTL, CTL$^*$ and CTL$^+$.
We 
compared their expressive power and provided tight bounds for their model checking and satisfiability problems. We further showed that, in most cases, these problems are not more costly than the corresponding problems for extended CTL or alternatively the corresponding problems for the base logics without parameters, despite the gain in expressive power. However, we also proved that the ability to use just two $\mathcal{F}^\mathcal{A}$ modalities 
leads to undecidability when the language parameter includes the deterministic context-free languages. Thus, the robust decidability of Extended CTL is not preserved when the base logic is more expressive.

Our comprehensive complexity analysis for both satisfiability and model checking explains why we did not restrict the language classes for the $\mathcal{U}$ and $\mathcal{R}$ modality separately as for Extended CTL in \cite{Axelsson2010}. Indeed, a close inspection of the proofs shows that hardness of the model checking problem holds already if either modality is equipped with the respective language class and the other with no class at all, except for the case of CTL$^+[VPL]$ on VPS/PDS. Thus, a more fine-grained analysis would not generally lead to improved complexity bounds. The same holds for the satisfiability problem except for CTL$^c$[VPL] where the complexity might change if weaker (or deterministic) classes are used for the release-operator. We leave these questions for future work.

For related work, we first mention the classical work extending temporal logics by regular operators \cite{Wolper1983,Kupferman2001,Giacomo2013} which are automata or regular expressions. Another approach are variants of Propositional Dynamic Logic \cite{Fischer1979,Streett1982,Lange2006}. These approaches have been extended to VPL \cite{Bozzelli2018,Weinert2018,Loeding2007}. However, none of these works analyse the model checking problem for different types of system models or provide a generic framework for different language classes. Furthermore, to our knowledge, CTL$^+$ has not yet been analysed with regard to formal language extensions and has not been checked against VPS or PDS.
Variants of CTL$^*$ and CTL$^+$ obtained by restricting the paths referenced by existential and universal quantifiers to sets defined by $\omega$-languages have been considered in \cite{Latte2014}. This approach differs from ours in many respects: their logics collapse to CTL$^*$ and CTL$^+$ when the same label is given for every transition while ours do not, they only consider the model checking problem for KT and they obtain decidable model checking problems for $\omega$-context-free languages, while the use of context-free languages leads to undecidability for our logics. 

	\bibliographystyle{acm}
	\bibliography{sections/conclusion/paper}
\end{document}